\documentclass[a4paper,UKenglish]{lipics-v2018}

\usepackage{amsmath}
\usepackage{listings}

\let\origthelstnumber\thelstnumber
\makeatletter
\newcommand*\Suppressnumber{%
  \lst@AddToHook{OnNewLine}{%
    \let\thelstnumber\relax%
     \advance\c@lstnumber-\@ne\relax%
    }%
}

\newcommand*\Reactivatenumber{%
  \lst@AddToHook{OnNewLine}{%
   \let\thelstnumber\origthelstnumber%
   \advance\c@lstnumber\@ne\relax}%
}
\makeatother

\usepackage{wrapfig}
\usepackage{multicol}
\newcommand{\listingcaption}[1]%
{%
\refstepcounter{lstlisting}\hfill%
Listing \thelstlisting: #1\hfill
\vspace{1em}
}%

\usepackage{ifthen}
\usepackage{color}
\usepackage{amssymb}

\definecolor{gpcolor}{rgb}{0.6,0.2,0.3}
\newboolean{showcomments}
\setboolean{showcomments}{false}
\ifthenelse{\boolean{showcomments}}
{ \newcommand{\mynote}[3]{
    \fbox{\bfseries\sffamily\scriptsize#1}
    {\small$\blacktriangleright$\textsf{\emph{\color{#3}{#2}}}$\blacktriangleleft$}}
}
{
\newcommand{\mynote}[3]{}}
\newcommand{\pk}[1]{\mynote{Petr}{#1}{gpcolor}}

\newcommand{\remove}[1]{}

\usepackage{graphicx}

\usepackage{comment}

\title{Parallel Combining: Benefits of Explicit Synchronization}
\titlerunning{Parallel Combining}

\author{Vitaly Aksenov}{ITMO University, Saint-Petersburg, Russia and Inria, Paris, France}{aksenov@corp.ifmo.ru}{}{}
\author{Petr Kuznetsov}{LTCI, T\'el\'ecom ParisTech, Universit\'e Paris-Saclay, Paris, France}{petr.kuznetsov@telecom-paristech.fr}{}{}
\author{Anatoly Shalyto}{ITMO University, Saint-Petersburg, Russia}{shalyto@mail.ifmo.ru}{}{}

\authorrunning{V. Aksenov, P. Kuznetsov and A. Shalyto}

\Copyright{V. Aksenov, P. Kuznetsov and A. Shalyto}

\subjclass{Computing methodologies $\rightarrow$   Concurrent
  computing methodologies, Parallel computing methodologies;Theory of computation $\rightarrow$ Models of computation
  $\rightarrow$ Concurrency $\rightarrow$ Distributed computing
  models, Parallel computing models}


\EventEditors{Jiannong Cao, Faith Ellen, Luis Rodrigues, and Bernardo Ferreira}
\EventNoEds{4}
\EventLongTitle{22st International Conference on Principles of Distributed Systems (OPODIS 2018)}
\EventShortTitle{OPODIS 2018}
\EventAcronym{OPODIS}
\EventYear{2018}
\EventDate{December 17--19, 2018}
\EventLocation{Hong Kong, China}
\EventLogo{}
\SeriesVolume{122}
\ArticleNo{0} 
\nolinenumbers 
\hideLIPIcs  

\begin{document}
\maketitle

\begin{abstract}
%
%
%
%

A \emph{parallel batched} data structure is designed to process
synchronized \emph{batches} of operations on the data structure
using a parallel program.
In this paper, we propose \emph{parallel combining}, a technique that
implements a \emph{concurrent} data structure
from a parallel batched one. 
The idea is that we explicitly synchronize concurrent operations
into batches: one of the processes becomes a \emph{combiner} which
collects concurrent requests 
and initiates a parallel batched algorithm
involving the owners (\emph{clients}) of the collected requests.
Intuitively, the cost of synchronizing the concurrent calls can be
compensated by running the parallel batched algorithm. 
%

We validate the intuition via  two applications. 
First, we use parallel combining to design
a concurrent data structure optimized for \emph{read-dominated}
workloads, 
taking a dynamic graph data structure as an example.
Second, we use a novel parallel batched \emph{priority queue} to
build a concurrent one. 
In both cases, we obtain performance gains with respect to
the state-of-the-art algorithms.  
\end{abstract}

\clearpage

\section{Introduction}

To ensure correctness of concurrent computations,  
various synchronization techniques are employed.
Informally, synchronization is used to handle conflicts on \emph{shared data}, e.g., resolving data races, or
\emph{shared resources}, e.g., allocating and deallocating memory.
%
Intuitively, the more sophisticated conflict patterns a concurrent program is
subject to---the higher are the incurred synchronization costs.  

Let us consider a concurrent-software class which we call
\emph{parallel programs}.
Provided an input, a parallel program aims at computing an output
that satisfies a \emph{specification}, i.e., an input-output relation (Figure~\ref{fig:parallel}). 
To boost performance, the program distributes the computation across
multiple parallel processes.  
Parallel programs are typically written for two environments:
for \emph{static multithreading} and \emph{dynamic multithreading}~\cite{thomas2009introduction}.
In \emph{static multithreading}, each process is given its own program and
these programs are written as a composition of \emph{supersteps}.
During a superstep, the processes perform conflict-free individual computations
and, when done, synchronize to accumulate the results.
In \emph{dynamic multithreading},
the program is written using dynamically called \emph{fork-join}
mechanisms (or similar ones, e.g., \texttt{\#pragma omp parallel} in
OpenMP~\cite{OpenMP}). 
In both settings, synchronization only appears in a specific form:
memory allocation/deallocation, and
aggregating superstep computations or
thread scheduling~\cite{frigo1998implementation}. 

\begin{figure}[h]
\centering
\begin{minipage}{0.48\textwidth}
  \centering
  \includegraphics[width=0.8\textwidth]{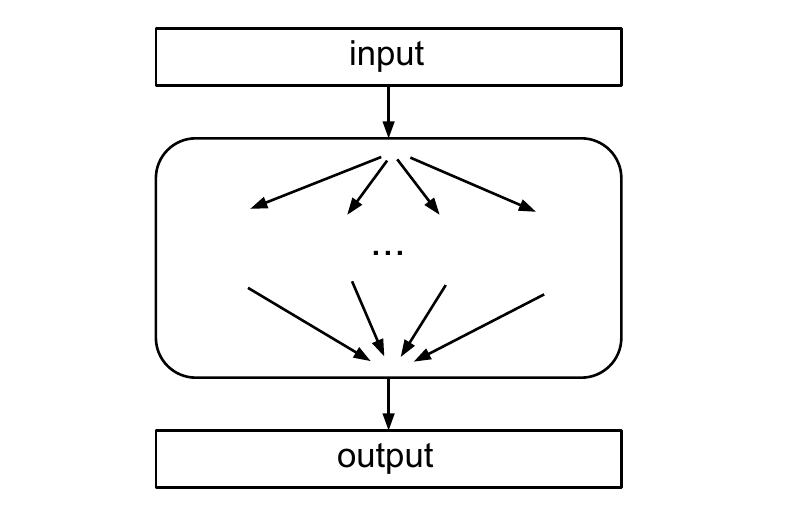}
  \caption{Execution of a parallel program}
  \label{fig:parallel}
\end{minipage}\hfill
\begin{minipage}{0.48\textwidth}
  \centering
  \includegraphics[width=0.8\textwidth]{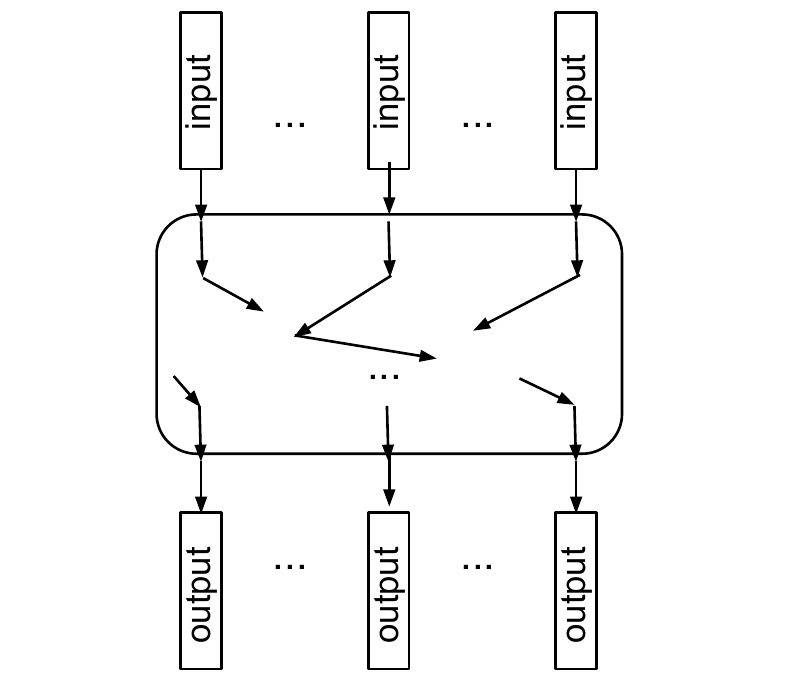}
  \caption{Execution on a concurrent data structure}
  \label{fig:concurrent}
\end{minipage}
\end{figure}

General-purpose \emph{concurrent data structures}, such as stacks,
binary search trees and priority queues, operate in a much less 
controlled environment.   
They are programmed to accept and process asynchronous operation
calls, which come from multiple concurrent processes and may interleave arbitrarily.
If we treat operation calls as inputs and their responses as outputs,
we can say that inputs and outputs are distributed across the processes
(Figure~\ref{fig:concurrent}).
It is typically expected that the interleaving operations
match the high-level sequential semantics of the data
type~\cite{HW90}, which is hard to implement efficiently given  
diverse and complicated data-race patterns often observed in
this kind of programs.
Therefore, designing efficient and correct concurrent data structures 
requires a lot of ingenuity from the programmer.
In particular, one should strive to provide the ``just right'' amount
of synchronization.    
\emph{Lock-based} data structures obviate data races by 
using fine-grained locking ensuring that contested data is
accessed in a mutually exclusive way.
\emph{Wait-free} and \emph{lock-free} data structures allow data races
but mitigate their effects by additional mechanisms, such
as \emph{helping} where one process may perform some work on behalf of other processes~\cite{herlihy2011art}.

As parallel programs are written for a restricted environment with
simple synchronization patterns, they are typically
easier to design than concurrent data structures.
%
In this paper, we suggest benefiting from this complexity gap by building concurrent data structures \emph{from}
parallel programs.
We describe a methodology of designing a concurrent data structure
from its \emph{parallel batched} counterpart~\cite{agrawal2014provably}.
A parallel batched data structure is a special case of a parallel program that accepts 
\emph{batches} (sets) of operations on a given sequential data type and executes them in a parallel way.
In our approach, we \emph{explicitly} synchronize concurrent operations,
assemble them into batches, and apply these batches on an \emph{emulated} parallel batched
data structure.

More precisely, concurrent processes share a set of \emph{active requests}
using any combining algorithm~\cite{oyama1999executing,hendler2010flat,
fatourou2011highly,fatourou2012revisiting}.
One of the processes with an active request becomes a \emph{combiner} and forms a
\emph{batch} from the requests in the set.
Under the coordination of the combiner, the owners of the collected requests, called \emph{clients},
apply the requests in the batch to the parallel batched data
structure.
As we show, this technique becomes handy when the overhead of \emph{explicitly
synchronizing} operation calls in batches
is compensated by the advantages of involving clients into the
computation using the parallel batched data structure.

\remove{
In the extreme case of ``concurrency-averse'' data structures, such as queues and stacks,
having control over the batch can be used to \emph{eliminate} certain requests
and bypass sequential bottlenecks,
even though there are no benefits of parallelizing the execution.\footnote{Note that sequential
combining~\cite{oyama1999executing,hendler2010flat,fatourou2011highly,fatourou2012revisiting,holt2013flat},
typically used in the concurrency-averse case, 
is a degenerate case of our parallel combining in which the combiner
applies the batch sequentially.} 
But as we show in the paper, combining \emph{and} parallel batching
pay off for data structures that offer some
degree of parallelism, such as dynamic graphs and priority queues.  
}

We discuss two applications of parallel combining and experimentally
validate performance gains.
First, we design concurrent implementations optimized for \emph{read-dominated} workloads
given a sequential data structure. Intuitively, updates are performed
sequentially and read-only operations are performed by the clients in parallel under the coordination of the combiner.
In our performance analysis, we considered a \emph{dynamic
  graph} data structure \cite{holm2001poly} that can be accessed for adding and
removing edges (updates), as well as for checking the connectivity between
pairs of vertices (read-only).
Second, we apply parallel combining to \emph{priority queue} that
is subject to sequential bottlenecks for minimal-element extractions,
while most insertions can be applied concurrently.
As a side contribution, we propose a novel parallel batched priority
queue, as no existing batched priority queue we are aware of can be
efficiently used in our context.
Our perfomance analysis shows that implementations based on parallel combining may outperform
state-of-the-art algorithms.

\subparagraph*{Structure.}
The rest of the paper is organized as follows.
In Section~\ref{sec:prel}, we give preliminary definitions.
In Section~\ref{sec:pc}, we outline parallel combining technique.
In Sections~\ref{sec:read-optimized} and~\ref{sec:priority-queue}, we present applications of our technique.
%
%
In Section~\ref{sec:exp}, we report on the outcomes of our performance analysis.
In Section~\ref{sec:related}, we overview the related work.
We conclude in Section~\ref{sec:discussion}.

\section{Background}
\label{sec:prel}

\subparagraph*{Data types and data structures.}
A sequential \emph{data type} is defined 
via a set of operations,
a set of responses, a set of states,
an initial state
and a set of transitions.
Each transition maps a state and an operation to a new state
and a response.
A \emph{sequential implementation} (or \emph{sequential data structure})
corresponding to a given data type
specifies, for each operation, a sequential read-write algorithm, so
that the specification of the data type is respected in every
sequential execution.   

We consider a system of $n$ asynchronous \emph{processes} (processors
or threads of computation) that communicate by 
performing primitive operations on shared \emph{base objects}. 
The primitive operations can be reads, writes, or 
\emph{conditional} operations, such as test\&set or compare\&swap.
A \emph{concurrent implementation} (or \emph{concurrent data
  structure}) of a given data type assigns, for each
process and each operation of the data type, a 
state
machine that is triggered whenever the process invokes an operation
and specifies the sequence of \emph{steps} (primitives on the base
objects) the process needs to perform to complete the operation.
We require the implementations to be \emph{linearizable}
with respect to the data type, 
i.e., we require that operations \emph{take effect} instantaneously within their
intervals~\cite{HW90}.


\subparagraph*{Batched data structures.}
A \emph{batched implementation} (or \emph{batched data structure})
of a data type exports one function \texttt{apply}.
This operation takes a \emph{batch} (set) of
data type operations as a parameter and returns responses for these
operations that are consistent with some sequential application of
the operations to the current state of the data structure, which is
updated accordingly. 
We also consider extensions of the definition where we
explicitly define the ``batched'' data type
via the \emph{set}~\cite{set-lin} or \emph{interval}~\cite{interval-lin}
linearizations.
Such a data type takes a batch and a state, and returns a new state and
a vector of responses.

For example, in the simplest form, a  batched implementation
may sequentially apply operations from a batch to
the sequential data structure.
But batched implementations may also use parallelism
to accelerate the execution of the batch:
we call these
\emph{parallel} batched implementations.
We consider two types of  parallel batched implementations:
\emph{static-multithreading} ones
and \emph{dynamic-multithreading} ones~\cite{thomas2009introduction}.
 
\subparagraph*{Static multithreading.} 
A parallel batched data structure specifies
a distinct (sequential) code to each process
in PRAM-like models (PRAM~\cite{jaja1992introduction},
Bulk synchronous parallel model~\cite{valiant1990bridging},
Asynchronous PRAM~\cite{gibbons1989more}, etc.).
For example, in this paper, we provide
a batched implementation of a priority queue in
the Asynchronous PRAM model.
The Asynchronous PRAM consists of $n$ sequential processes,
each with its own private local memory, communicating
through the shared memory. Each process has its own program.
Unlike the classical PRAM model, each process executes its instructions independently
of the timing of the other processors. Each process performs one of the
four types of instructions per tick of its local clock:
global read, global write, local operation, or
synchronization step. A synchronization step for a set $S$ of processes
is a logical point where each processor in $S$ waits for all the processes in $S$
to arrive before continuing its local program.

\subparagraph*{Dynamic multithreading.}
Here the parallel batched implementation is written as a sequential read-write algorithm
using concurrency keywords specifiying logical parallelism,
such as \emph{fork}, \emph{join} and \emph{parallel-for} \cite{thomas2009introducction}.
An execution of a batch can be presented as a directed acyclic graph (\emph{DAG}) that unfolds dynamically. 
In the DAG, nodes represent unit-time sequential subcomputations, and 
edges represent control-flow dependencies between nodes. A node that corresponds
to a ``fork'' has two or more outgoing edges and a node that corresponds to a ``join'' has two
or more incoming edges.
The batch is executed using a \emph{scheduler} that chooses which
DAG nodes to execute on each process.
It can only execute \emph{ready} nodes: not yet executed nodes whose predecessors have all been
executed.
The most commonly used \emph{work-stealing} scheduler
(e.g.,~\cite{blumofe1999scheduling}) operates as follows.
Each process $p$ is provided with a deque for ready nodes.
When process $p$ completes node $u$, it traverses successors of $u$ and collects
the ready ones. Then $p$ selects one of the ready successors for execution and
adds the remaining ready successors at the bottom of its deque. When $p$'s deque is empty, it becomes
a \emph{thief}: it randomly picks a victim processor and steals from the top of the victim's
deque.

\section{Parallel Combining}
\label{sec:pc}

%
%
In this section, we describe the \emph{parallel combining} technique
in a parameterized form:
the parameters are specified depending on the application.
We then discuss how to use the technique in transforming parallel
batched programs into concurrent data structures. 

\subsection{Combining Data Structure}

Our technique relies on a \emph{combining} data structure
$\mathbb{C}$ (e.g., the one used in~\cite{hendler2010flat}) that
maintains a set of requests to a data structure
and determines which process 
is a combiner.
If the set of requests is not empty then
exactly one process should be a combiner.

Elements stored in $\mathbb{C}$ are of \texttt{Request} type
consisting of the following fields:
1)~the method to be called and its input;
2)~the response field;
3)~the status of the request with a value in an application-specific \texttt{STATUS\_SET};
4)~application-specific auxiliary fields.
In our applications, \texttt{STATUS\_SET} contains, by default,
values \texttt{INITIAL} and \texttt{FINISHED}: \texttt{INITIAL}
meaning that the request is in the initial state, and 
\texttt{FINISHED} meaning that the request is served.

$\mathbb{C}$ supports three operations:
1)~\texttt{addRequest(r : Request)} inserts request \texttt{r}
  into the set, and the response indicates whether the calling process becomes a combiner
  or a client;
2)~\texttt{getRequests()} returns a non-empty set of requests; and
3)~\texttt{release()} is issued by the combiner to make
  $\mathbb{C}$ find another process to be a combiner.

In the following, we use any black-box implementation of
$\mathbb{C}$ providing this functionality~\cite{oyama1999executing,hendler2010flat,fatourou2011highly,fatourou2012revisiting}.

\subsection{Specifying Parameters}
To perform an operation, a process executes the following steps
(Figure~\ref{listing:combining}):
1)~it prepares a request, inserts the request into $\mathbb{C}$ using \texttt{addRequest($\cdot$)};
2)~if the process becomes the combiner (i.e.,
\texttt{addRequest}($\cdot$) returned $\textit{true}$), 
it collects requests from $\mathbb{C}$ using \texttt{getRequests()},
then it executes algorithm \texttt{COMBINER\_CODE},
and, finally, 
it calls \texttt{release()} to enable another active process to become a combiner;
3)~if the process is a client (i.e.,
\texttt{addRequest}($\cdot$) returned $\textit{false}$),
  it waits until the status of the request
becomes not \texttt{INITIAL} and, then, executes algorithm \texttt{CLIENT\_CODE}.

\begin{figure}
\begin{lstlisting}[
%caption={Parallel combining: pseudocode},
%captionpos=b,
%label=listing:combining,
%multicols=2
]
Request:
  method
  input
  res
  status $\in$ STATUS_SET
  ...

$\textbf{execute(method, input)}$:
  req $\leftarrow$ new Request()       |\label{line:combining:init:start}|
  req.method $\leftarrow$ method
  req.input $\leftarrow$ input
  req.status $\leftarrow$ INITIAL   |\label{line:combining:init:end}|
  if  $\mathbb{C}.$addRequest(req): 
    // combiner
    $A \leftarrow $ $\mathbb{C}.$getRequests() 
    COMBINER_CODE
    $\mathbb{C}$.release()
  else:
    while req.status $=$ INITIAL:
      nop
    CLIENT_CODE 
  return
\end{lstlisting}
\caption{Parallel combining: pseudocode}
\label{listing:combining}
\end{figure}

To use our technique, one should therefore specify 
\texttt{COMBINER\_CODE},  
\texttt{CLIENT\_CODE},
and appropriately modify \texttt{Request} type and \texttt{STATUS\_SET}.

Note that \emph{sequential combining}~\cite{oyama1999executing,hendler2010flat,fatourou2012revisiting,drachsler2014lcd}
is a special case of parallel combining in which all the work is done
by the combiner, and the client code is empty.

\subsection{Parallel Batched Algorithms}
\label{subsec:batch}
We discuss how to build a concurrent data structure
given a parallel batched one in one of two forms: for static or
dynamic multithreading.

In the static multithreading case, each process is provided with a distinct version
of \texttt{apply} function.
%
%
We enrich \texttt{STATUS\_SET} with \texttt{STARTED}.
In \texttt{COMBINER\_CODE}, the combiner collects the requests,
sets their status to \texttt{STARTED},
performs the code of \texttt{apply}
and waits for the clients to become \texttt{FINISHED}.
In \texttt{CLIENT\_CODE} the client waits until its request has
\texttt{STARTED} status, performs
the code of \texttt{apply} and sets the status of its request to \texttt{FINISHED}.

Suppose that we are given a
parallel batched implementation for dynamic multithreading.
One can turn it into a concurrent one
using parallel combining with the \emph{work-stealing} scheduler.  
Again, we enrich \texttt{STATUS\_SET} with \texttt{STARTED}.
In \texttt{COMBINER\_CODE}, the combiner collects the requests
and sets their status to \texttt{STARTED}.
Then the combiner creates a working deque, puts there 
a new node of computational DAG with \texttt{apply} function
and starts the work-stealing routine on processes-clients.
Finally, the combiner waits for the clients to become \texttt{FINISHED}.
In \texttt{CLIENT\_CODE}, the client creates a working deque
and starts the work-stealing routine.

In Section~\ref{sec:priority-queue}, we illustrate the use of parallel combining and parallel batched
programs on the example of a priority queue.


\section{Read-Optimized Concurrent Data Structures}
\label{sec:read-optimized}
Before discussing parallel batched algorithms, let us consider a
natural application of parallel combining: data structures 
optimized for \emph{read-dominated} workloads.
%

Suppose that we are given a sequential data structure $D$ that supports
\emph{read-only} (not modifying the data structure) operations, the
remaining operations  are called 
\emph{updates}.
%
We assume a scenario where read-only operations
dominate over other updates.
%

Now we explain how to set parameters of parallel combining for this application.
At first, \texttt{STATUS\_SET} consists of three elements \texttt{INITIAL}, \texttt{STARTED} and \texttt{FINISHED}.
\texttt{Request} type does not have auxiliary fields.

\begin{figure}[t]
\begin{lstlisting}[
%caption={Parallel combining applied to read-optimized data structures: COMBINER\_CODE},
%captionpos=b,
%label=listing:rw,
%multicols=2
]
COMBINER_CODE:  |\label{line:combiner}|
  R $\leftarrow \emptyset$ |\label{line:combiner:split:start}|

  for r $\in$ A:
    if isUpdate(r.method):
      apply(D, r.method, r.input)
      r.status $\leftarrow$ FINISHED
    else:
      R $\leftarrow$ R $\cup$ r |\label{line:combiner:split:end}|

  for r $\in$ R:                      |\label{line:combiner:started:start}|
    r.status $\leftarrow$ STARTED    |\label{line:combiner:started:end}|
  if req.status = STARTED:              |\label{line:combiner:read:check}|
    apply(D, req.method, req.input)       |\label{line:combiner:read:start}|
    req.status $\leftarrow$ FINISHED   |\label{line:combiner:read:end}|

  for $r \in R$:                   |\label{line:combiner:finished:start}|
    while r.status $=$ STARTED:  
      nop                          |\label{line:combiner:end}|

CLIENT_CODE:  |\label{line:client}|
  if not isUpdate(req.method):
    apply(D, req.method, req.input)        |\label{line:client:execute}|
    req.status $\leftarrow$ FINISHED    |\label{line:client:end}|
\end{lstlisting}
\caption{Parallel combining in application to read-optimized data structures}
\label{listing:rw}
\end{figure}

In \texttt{COMBINER\_CODE} (Figure~\ref{listing:rw} Lines~\ref{line:combiner}-\ref{line:combiner:end}),
the combiner 
iterates through the set of collected requests $A$: if a request contains an update operation then
the combiner executes it and sets its status to \texttt{FINISHED};
otherwise, the combiner adds the request to set $R$.
Then the combiner sets the status of requests in $R$ to \texttt{STARTED}.
After that the combiner checks whether its own request is read-only.
If so, it executes the method and sets the status of its request to \texttt{FINISHED}.
Finally, the combiner waits until the status of the requests in $R$
become \texttt{FINISHED}.
                         
In \texttt{CLIENT\_CODE} (Figure~\ref{listing:rw} Lines~\ref{line:client}-\ref{line:client:end}),
the client checks whether its method is read-only.
If so, the client executes the method and sets
the status of the request to \texttt{FINISHED}.


\begin{theorem}
\label{theorem:rw}
Algorithm in Figure~\ref{listing:rw} produces a
linearizable concurrent data structure from a sequential one.
\end{theorem}
\begin{proof}
Any execution of the algorithm can be viewed as a series of non-overlapping combining phases
(Figure~\ref{listing:rw}, Lines~\ref{line:combiner:split:start}-\ref{line:combiner:end}).
We can group the operations into batches by
the combining phase in which they are applied.

%
Each update operation is linearized at the point when
the combiner applies this operation.
Note that this is a correct linearization since all operations that are linearized
before are already applied: the operations from preceding combining phases
were applied during the preceding phases, while the operations
from the current combining phase are applied sequentially by the combiner.

%
Each read-only operation is linearized at the point
when the combiner sets the status of the corresponding request
to \texttt{STARTED}.
By the algorithm, a read-only operation observes all update operations
that are applied before and during the current combining phase.
Thus, the chosen linearization is correct.
\end{proof}
 
To evaluate the approach in practice we implement a concurrent \emph{dynamic graph} data structure
by Holm et al. and execute it in read-dominated environments~\cite{holm2001poly} (Section~\ref{sec:dynamic:graph}).

\section{Priority Queue}
\label{sec:priority-queue}

\emph{Priority queue} is an abstract data type that maintains
an ordered multiset and supports two operations:
\begin{itemize}
\item \texttt{Insert}($v$)~--- inserts value $v$ into the set;
\item $v \leftarrow$ \texttt{ExtractMin}()~--- extracts the smallest value from the set.
\end{itemize}

To the best of our knowledge, no prior parallel implementation
of a priority-queue~\cite{pinotti1991parallel,deo1992parallel,brodal1998parallel,sanders1998randomized}
can be efficiently used in our context: their complexity inherently
depends on the total number of processes in the system, regardless of
the actual batch size.
We therefore introduce a novel heap-based parallel batched priority-queue implementation
in a form of \texttt{COMBINER\_CODE} and \texttt{CLIENT\_CODE}
convenient for parallel combining.
The concurrent priority queue is then derived from the described below
parallel batched one using the approach presented in Section~\ref{subsec:batch}.  

Here we give only the brief overview of our parallel batched algorithm.
%
Please refer to Appendix~\ref{sec:priority} for a detailed description.

In Section~\ref{sec:exp:priority-queue}, we show that the resulting
concurrent priority queue is able to outperform manually crafted
state-of-the-art implementations.

\subsection{Sequential Binary Heap}

Our batched priority queue
is based on the \emph{sequential binary heap}
by Gonnet and Munro \cite{gonnet1986heaps},
one of the simplest and fastest sequential priority queues.
We briefly describe this algorithm below.

A binary heap of size $m$ is represented as a complete binary tree
with nodes indexed by $1, \ldots, m$.
Each node $v$ has at most two children: $2v$ and $2v + 1$ (to exist, $2v$ and $2v + 1$
should be less than or equal to $m$).
For each node, the \emph{heap property} should be satified:
the value stored at the node is less than the values stored at its children.

The heap is represented with \emph{size} $m$ and an array $a$
where $a[v]$ is the value at node $v$.
Operations \texttt{ExtractMin} and \texttt{Insert} are performed as follows:
\begin{itemize}
\item \texttt{ExtractMin} records the value $a[1]$ as a response,
  copies $a[m]$ to $a[1]$,
  decrements $m$ and performs the \textit{sift down} procedure to restore the heap property.
  Starting from the root, for each node $v$ on the path,
  we check whether value $a[v]$ is less than values
  $a[2v]$ and $a[2v + 1]$. If so, then the heap property is satisfied and
  we stop the operation.
  Otherwise, we choose the child $c$, either $2v$ or $2v + 1$, with the smallest value, swap values $a[v]$
  and $a[c]$, and continue with $c$.
\item \texttt{Insert($x$)} initializes a variable $val$ to $x$, increments $m$
  and traverses the path from the root to a new node $m$.
  For each node $v$ on the path, if $val<a[v]$, then the two values
  are swapped.
  Then the operation continues with the child of $v$ that lies on the path
  from $v$ to node $m$.
  Reaching node $m$ the operation sets its value to $val$.
\end{itemize}
The complexity is $O(\log m)$ steps per operation.

\subsection{Setup}
The heap is defined by its size $m$ and an array $a$ of Node objects.
Node object has two fields: value $val$ and boolean $locked$ (In Appendix~\ref{sec:priority} it has an additional field).

\texttt{STATUS\_SET} consists of three items: \texttt{INITIAL}, \texttt{SIFT} and \texttt{FINISHED}.

A Request object consists of: a method $method$ to be called and its input argument $v$;
a result $res$ field; a $status$ field and a node identifier $start$.

\subsection{ExtractMin Phase}

\subparagraph*{Combiner: ExtractMin preparation}
The combiner withdraws requests $A$ from combining data structure $\mathbb{C}$.
It splits $A$ into sets $E$ and $I$: the set of ExtractMin requests and
Insert requests.
Then it finds $|E|$ nodes $v_1, \ldots, v_{|E|}$ of heap with
the smallest values using the Dijkstra-like algorithm in $O(|E| \cdot \log |E|)$ steps:
(i)~create a heap of nodes ordered by values,
put there the root $1$;
(ii)~at each of the next $|E|$ steps withdraw the node $v$ with the
minimal value from the heap;
(iii)~put two children of $v$, $2v$ and $2v+1$, to the heap.
The $|E|$ withdrawn nodes are the nodes with the $|E|$ minimal
values.
For each request $E[i]$, the combiner sets $E[i].res$ to $a[v_i].val$, $a[v_i].locked$ to \texttt{true},
and $E[i].start$ to $v_i$.

The combiner proceeds by pairing Insert requests in $I$
with ExtractMin requests in $E$ using the following procedure.
Suppose that $\ell = \min(|E|, |I|)$.
For each $i \in [1, \ell]$, the combiner sets $a[v_i].val$ to $I[i].v$
and $I[i].status$ to \texttt{FINISHED},
i.e., this Insert request becomes completed.
Then, for each $i \in [\ell + 1, |E|]$, the combiner sets $a[v_i].val$ to the value of the last node $a[m]$ and
decrements $m$, as in the sequential algorithm.
Finally, the combiner sets the status of all requests in $E$ to \texttt{SIFT}.

\subparagraph*{Clients: ExtractMin phase}
Briefly, the clients sift down the values in nodes $v_1, \ldots, v_{|E|}$
in parallel using hand-over-hand locking:
the $locked$ field of a node is set whenever there is a \emph{sift down} operation
working on that node.

A client $c$ waits until the status of its request becomes \texttt{SIFT}.
$c$ starts sifting down from $req.start$.
Suppose that $c$ is currently at node $v$.
$c$ waits until the $locked$ fields of the children become \texttt{false}.
If $a[v].val$, the value of $v$, is less than the values
in its children, then \emph{sift down} is finished:
$c$ unsets $a[v].locked$ and sets the status of its request to \texttt{FINISHED}.
Otherwise, let $w$ be the child with the smallest value.
Then $c$ swaps $a[v].val$ and $a[w].val$, sets $a[w].locked$,
unsets $a[v].locked$ and continues with node $w$.

If the request of the combiner is ExtractMin, it also runs the code above as a client.
The combiner considers the ExtractMin phase completed when all requests in $E$ have status \texttt{FINISHED}.

\subsection{Insert Phase}
For simplicity, we describe first the sequential algorithm.

At first, the combiner removes all completed requests from $I$.
Then it initializes new nodes $m + 1, \ldots, m + |I|$ which we call
\emph{target nodes} and increments $m$ by $|I|$.
The nodes for which the subtrees of both children
contain at least one target node are called \emph{split nodes}.
(See Figure~\ref{fig:tree-short} for an example of how target and split
nodes can be defined.)

\begin{wrapfigure}{L}{0.5\textwidth}
\begin{center}
\includegraphics[width=6cm]{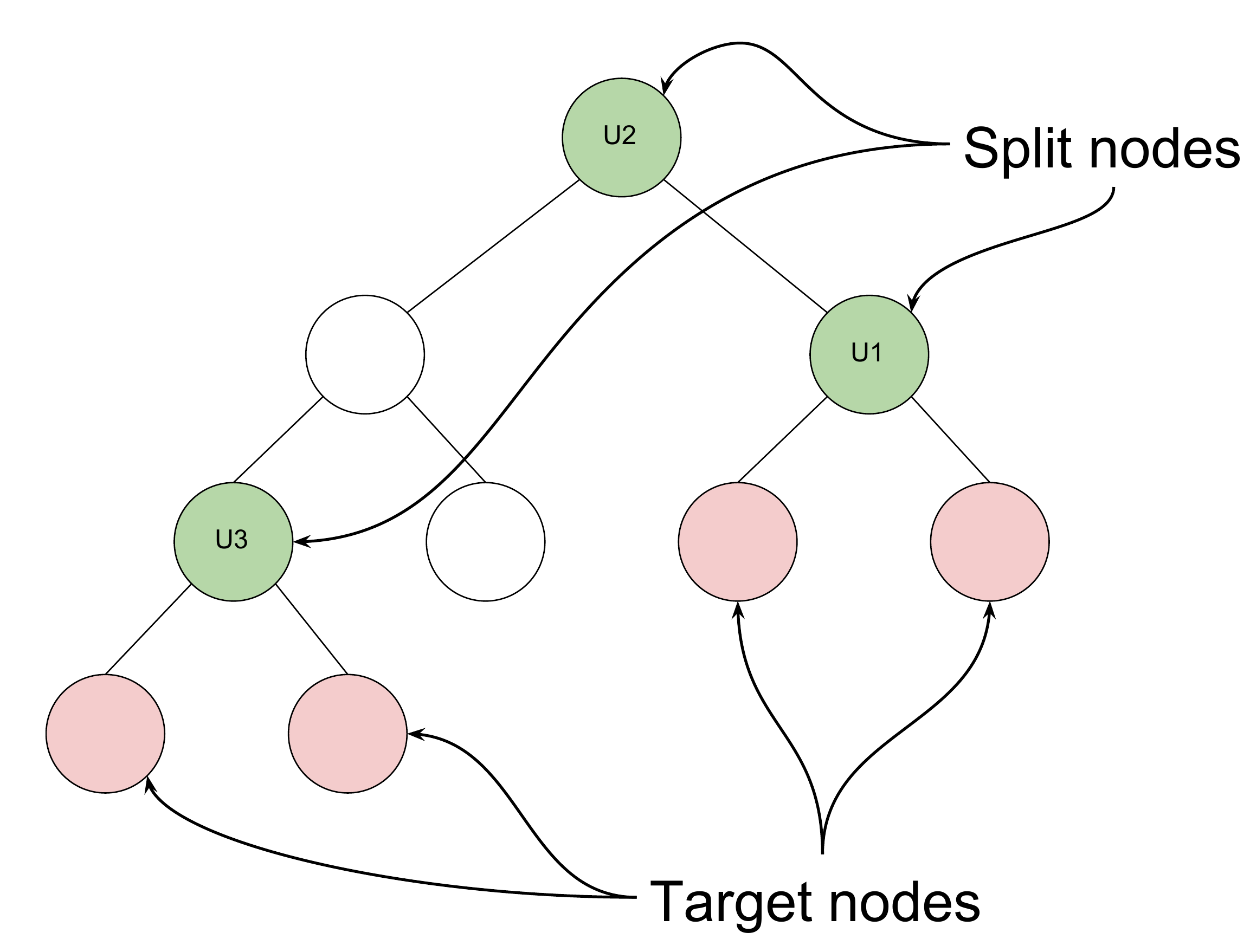}
\end{center}
\caption{Split and target nodes}
\label{fig:tree-short}
\end{wrapfigure}

The combiner collects the values of the remaining Insert requests and sorts them: $r_1, \ldots, r_{|I|}$.
Then it sets the status of these requests to \texttt{FINISHED}.

Now, we introduce InsertSet class: it consists of two sorted lists $A$ and $B$.
The combiner starts the following recursive procedure
at the root with InsertSet $s$: $s.A$ contains $r_1, \ldots, r_{|I|}$ while
$s.B$ is empty.
Suppose that the procedure is called on node $v$ and InsertSet $s$.
Let $min$ be the minimum out of the first element of $s.A$ and the first element of $s.B$.
If $v$ is a target node then the combiner sets $a[v].res$ to $min$ and withdraws $m$
from the corresponding list.
Otherwise, the combiner compares $a[v].res$ with $min$:
if $a[v].res$ is smaller, then it does nothing; otherwise, it appends $a[v].res$ to the end of $s.B$
(note that $s.B$ remains sorted because $s.B$ consists only of values that were ancestors in the heap),
withdraws $min$ from the corresponding list and sets $a[v].res$ to $min$.

If $v$ is not the split node the combiner calls the recursive procedure on the child with
target nodes in the subtree and with InsertSet $s$.
Otherwise, the combiner calculates $inL$ and $inR$~--- the number of target nodes in the left and
right subtrees of $v$.
Suppose, for simplicity, that $inL$ is less than $inR$ (the opposite case can be resolved similarly).
The combiner splits $s$ into two parts: create InsertSet $s_L$, move $\min(inL, |s.A|)$ first values from $s.A$ to $s_L.A$
and move $\min(inL - |s_L.A|, |s.B|)$ first values from $s.B$ to $s_L.B$.
Finally, it calls the recursive procedure on the left child with InsertSet $s_L$ and
on the right child with InsertSet $s$.

This algorithm works in $O(\log m + c \log c)$ steps (and can be optimized to $O(\log m + c)$ steps),
where $m$ is the size of the queue and $c$ is the number of Insert requests to apply.
Note that this algorithm is almost non-parallelizable due to its small complexity,
and our parallel algorithm is only developed to reduce constant factors.

%
Now, we construct a parallel algorithm for the Insert phase.
We enrich Node object with the IntegerSet field $split$.
The combiner sets the $start$ field of the first client ($i[1].start$) to the root $1$, while $start$ fields of other clients
to the right children of split nodes (we have exactly $|I| - 1$ split nodes).
Then it intializes the $split$ field of the root as the IntegerSet $s$ is initialized at the beginning of the sequential algorithm:
list $A$ contains values of requests while list $B$ is empty.

Each client waits until the $split$ field of the corresponding $start$ node is non-null.
Then it reads this IntegerSet: the values from this set should be inserted in the subtree.
Finally, the client performs the procedure similar to the recursive procedure from the sequential algorithm except for
one difference: when it reaches a split node instead of going recursively to left and right children,
it splits InsertSet to $s_L$ and $s_R$ of sizes $inL$ and $inR$,
puts $s_R$ into the $split$ field of the right child (in order to wake another client) and continues
with the left child and $s_L$. 

For further details about the parallel algorithm we refer to Appendix~\ref{sec:priority}.

\section{Experiments}
\label{sec:exp}
We evaluate Java implementations of our data structures 
on a 4-processor AMD Opteron 6378 2.4 GHz server with
16 threads per processor (yielding 64 threads in total), 512 Gb of RAM, running
Ubuntu 14.04.5 with Java 1.8.0\_111-b14 and HotSpot JVM 25.111-b14.


\subsection{Concurrent Dynamic Graph}
\label{sec:dynamic:graph}
To illustrate how parallel combining can be used to construct
read-optimized concurrent data structures,
we took the sequential dynamic graph implementation
by Holm et al.~\cite{holm2001poly}.
This data structure supports two update methods:
an insertion of an edge and a deletion of an edge;
and one read-only method:
a connectivity query that tests whether two vertices are
connected.

We compare our implementation based on parallel combining (PC)
with \emph{flat combining}~\cite{hendler2010flat} as a combining data structure
against three others:
(1)~Lock, based on ReentrantLock from \textsf{java.util.concurrent};
(2)~RW Lock, based on ReentrantReadWriteLock
from \textsf{java.util.concurrent}; and
(3)~FC, based on \emph{flat combining}~\cite{hendler2010flat}.
The code is available at \url{https://github.com/Aksenov239/concurrent-graph}.

We consider workloads parametrized with:
1)~the fraction $x$ of connectivity queries
  (50\%, 80\% or 100\%, as we consider read-dominated workloads);
2)~the set of edges $E$: edges of a single random tree, or
  edges of ten  random trees;
3)~the number of processes $P$ (from $1$ to $64$).
We prepopulate the graph on $10^5$ vertices with edges from $E$: we insert each edge
with probability $\frac{1}{2}$. Then we start $P$ processes.
Each process repeatedly performs operations:
1)~with probability $x$, it calls a connectivity query on two vertices chosen uniformly at random;
2)~with probability $1 - \frac{x}{2}$, it inserts an edge chosen
uniformly at random from $E$;
3)~with probability $1 - \frac{x}{2}$, it deletes an edge chosen
uniformly at random from $E$.

We denote the workloads with $E$ as a single tree as \emph{Tree} workloads,
and other workloads as \emph{Trees} workloads.
\emph{Tree} workloads are interesting because they show the degenerate case:
the dynamic graph behaves as a dynamic tree. 
In this case, about $50\%$ of update operations
successfully change the spanning fprest,
while other update operations only check
the existence of the edge and do not modify
the graph.
\emph{Trees} workloads are interesting because
a reasonably small number (approximately, 5-10\%)
of update operations        
modify the maintained set of all edges and
the underlying complex data structure that maintains
a spanning forest (giving in total the squared logarithmic complexity),
while other update operations can only modify
the set of edges
but cannot modify the underlying complex
data structure (giving in total the logarithmic complexity).

For each setting and each algorithm,
we run the corresponding workload for $10$ seconds to warmup HotSpot JVM
and then we run the workload five more times for $10$ seconds.
The average throughput of the last five runs is reported in Figure~\ref{fig:graph}.

\begin{figure}
\centering
\includegraphics[width=\textwidth]{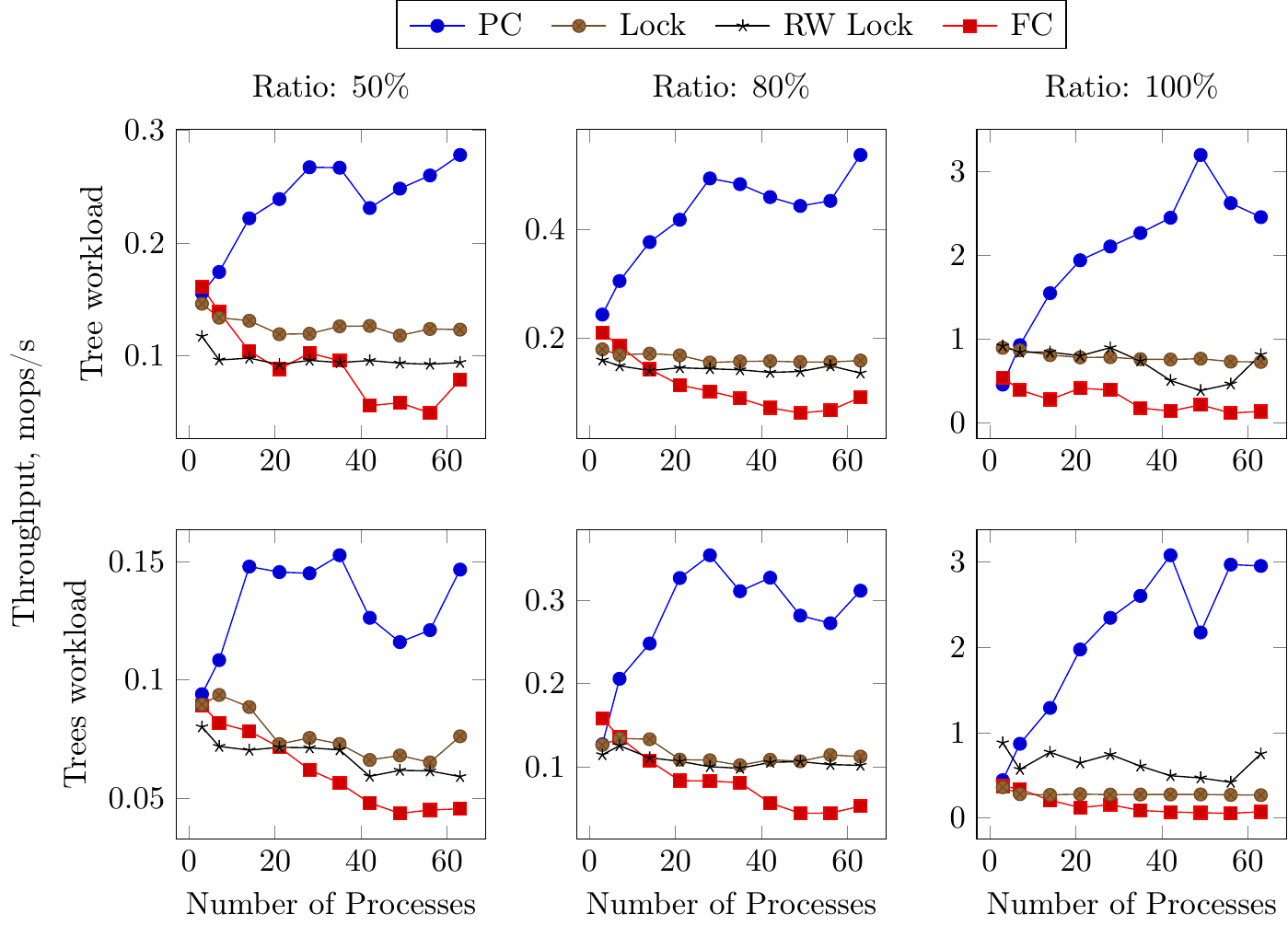}
\caption{Dynamic graph implementations}
\label{fig:graph}
\end{figure}


From the plots we can infer two general observations: PC exhibits the highest thoughput over all considered implementations
and it is the only one whose throughput
scales up with the number of the processes.
On the 100\% workload we expect the throughput curve to be almost linear since all operations
are read-only and can run in parallel.
The plots almost confirm our expectation: the curve of the throughput is a linear function
with coefficient $\frac{1}{2}$ (instead of the ideal coefficient $1$).
We note that this is almost the best we can achieve: a combiner typically collects
operations of only approximately half the number of working processes.
In addition, the induced overhead is still perceptible, since each connectivity
query works in just logarithmic time.
With the decrease of the fraction of read-only operations we expect that
the throughput curve becomes flatter, as 
plots for the 50\% and 80\% workloads confirm.  

It is also interesting to point out several features of other implementations.
At first, FC implementation works slightly worse than Lock and RW Lock.
This might be explained as follows.
Lock implementations (\texttt{ReentrantLock} and
\texttt{ReentrantReadWriteLock}) behind Lock and RWLock
implementations are based on CLH Lock~\cite{craig1993building}
organized as a \emph{queue}:
every competing process is appended to the queue and then waits until
the previous one releases the lock.
Operations on the dynamic graph take significant amount of time, so
under high load when the process finishes its operation it appends itself to
the queue in the lock without any contention.
Indeed, all other processes are likely to be in the queue and,
thus, no process can contend.
By that the operations by processes are serialized with almost no overhead.
%
In contrast, the combining procedure in FC introduces non-negligible
overhead related to gathering the requests and writing them into requests structures.

Second, it is interesting to observe that, against the intuition, RWLock is not so
superior with respect to Lock on read-only workloads.
As can be seen, when there are update operations in the workload
RWLock works even worse than Lock.
We relate this to the fact that the overhead hidden inside \texttt{ReentrantReadWriteLock}
spent on manipulation with read and write requests
is bigger than the overhead spent by \texttt{ReentrantLock}.
With the increase of the percentage of read-only operations the difference
between Lock and RWLock diminishes and RWLock becomes dominant
since read operations become more likely to be applied concurrently
(for example, on 50\% it is normal to have an execution without any parallelization:
read operation, write operation, read operation, and so on).
However, on 100\% one could expect that RWLock should exhibit ideal throughput.
Unfortunately, in this case, under the hood \texttt{ReentrantReadWriteLock} uses
compare\&swap on the shared variable that represents the number of current read operations.
Read-only operations take enough time but not enough to amortize the considerable traffic
introduced by concurrent compare\&swaps.
Thus, the plot for RWLock is almost flat, getting even slightly worse
with the increase of the number of processes,
and we blame the traffic for this.

\subsection{Priority Queue}
\label{sec:exp:priority-queue}

We run our algorithm (PC) with \emph{flat combining}~\cite{hendler2010flat}
as a combining data structure
against six state-of-the-art concurrent
priority queues: 
(1)~the lock-free skip-list by Linden and Johnson (Linden SL~\cite{linden2013skiplist}),
(2)~the lazy lock-based skip-list (Lazy SL~\cite{herlihy2011art}),
(3)~the non-linearizable lock-free skip-list by Herlihy and Shavit (SkipQueue~\cite{herlihy2011art})
as an adaptation of Lotan and Shavit's algorithm~\cite{shavit2000skiplist},
(4)~the lock-free skip-list from Java library (JavaLib),
(5)~the binary heap with flat combining (FC Binary~\cite{hendler2010flat}),
and (6)~the pairing heap with flat combining (FC Pairing~\cite{hendler2010flat}).
\footnote{We are aware of the cache-friendly priority queue by Braginsky et al.~\cite{braginsky2016cbpq},
but we do not have its Java implementation.}
The code is available at \url{https://github.com/Aksenov239/FC-heap}.

We consider workloads parametrized by:
1)~the initial size of the queue $S$ ($8\cdot 10^5$ or $8 \cdot 10^6$); and
2)~the number $P$ of working processes (from $1$ to $64$). 
We prepopulate the queue with $S$ random integers chosen uniformly
from the range $[0, 2^{31} - 1]$.
Then we start $P$ processes, and each process repeatedly  
performs operations:
with equal probability it either inserts a random value
taken uniformly from $[0, 2^{31} - 1]$
or extracts the minimum value.

For each setting and each algorithm, we run the corresponding
workload for $10$ seconds to warmup HotSpot JVM
and then we run the workload five more times for $10$ seconds.
The average throughput of the last five runs is reported in Figure~\ref{fig:heap}.

\begin{figure}
\centering
\includegraphics[width=0.7\textwidth]{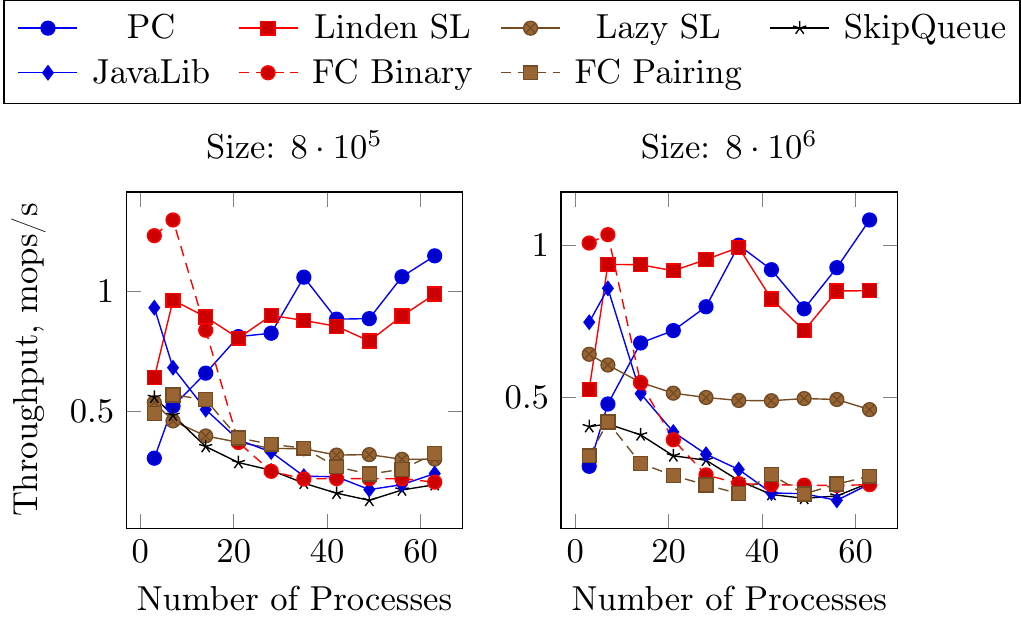}
\caption{Priority Queue implementations}
\label{fig:heap}
\end{figure}

%
%
%
%
On a small number of processes ($<15$), PC performs worse than other algorithms. 
With respect to Linden SL, Lazy SL, SkipQueue and JavaLib this can be explained by two different issues:
\begin{itemize}
\item Synchronization incurred by PC is not compensated by the work done;
\item Typically, a combiner collects operations of only approximately half the processes,
  thus, we ``overserialize'', i.e., only $\frac{n}{2}$ operations can be performed in parallel.
\end{itemize}
In contrast, on small number of processes, the other four algorithms can perform operations almost with no contention.
With respect to algorithms based on flat combining, FC Binary and FC Pairing,
our algorithm is simply slower on one process than
the simplest sequential binary and pairing heap algorithms.

With the increase of the number of processes
the synchronization overhead significantly increases
for all algorithms
(in addition to the fact that FC Binary and FC Pairing cannot scale).
As a result, starting from $15$ processes, PC outperforms all algorithms except for Linden SL.
Linden SL relaxes the contention during ExtractMin operations,
and it helps to keep the throughput approximately constant.
At approximately $40$ processes the benefits of the parallel
batched algorithm in PC starts prevailing the costs of explicit
synchronization, and our algorithms overtakes Linden SL.
%

It is interesting to note that FC Binary performs very well when the number of processes is small:
the overhead on the synchronization is very small, the processes are from the same core
and the simplest binary heap performs operations very fast.

\section{Related Work}                         
\label{sec:related}
To the best of our knowledge, Yew et al.~\cite{yew1987distributing}
were the first to propose combining concurrent operations.
They introduced a \emph{combining tree}: processes start at distinct leaves,
traverse upwards, and gain exclusive access by reaching the root.
If, during the traversal, two processes access the same
tree node, one of them adopts the operations of another and continues the
traversal, while the other 
waits until its
operations are completed. Several improvements of this technique have
been discussed, such as adaptive combining tree~\cite{shavit2000combining},
barrier implementations~\cite{gupta1989scalable,mellor1991algorithms}
and counting networks~\cite{shavit1996diffracting}.

A different approach was proposed by Oyama et al.~\cite{oyama1999executing}.
Here the data structure is protected by a lock. A thread with a new 
operation to be performed adds it to a list of submitted
requests and then tries to acquire the lock.
The 
winner of the lock performs the pending requests on
behalf of other processes from the list in LIFO order. 
The main drawback of this approach is that all processes
have to perform CAS on the head of the list.
The \emph{flat combining} technique presented by Hendler et
al. \cite{hendler2010flat} addresses this issue by 
replacing the list of requests with a \emph{publication list} which
maintains a distinct \emph{publication record} per participating process.
A process puts its new operation in its publication
record, and the publication record is only maintained
in the list if the process is ``active enough''.  
This way the processes generally do not contend on the head of the list.
Variations of flat combining
were later proposed for various contexts~\cite{fatourou2011highly,fatourou2012revisiting,holt2013flat,drachsler2014lcd}.

\emph{Hierarchical combining}~\cite{hendler2010scalable} is the first
attempt to improve performance of combining
using the computational power of clients.
The list of requests is split into blocks, and each of these blocks
has its own combiner.
The combiners push the combined requests from the block into the
second layer implemented as the standard flat combining with one
combiner.
%
This approach, however, may be sub-optimal as it does not involve \emph{all} 
clients.
Moreover, this approach works only for specific data structures,
such as stacks or unfair synchronous queues,
where operations could be combined without accessing the data structure.

In a different context, Agrawal et al.~\cite{agrawal2014provably} 
suggested to use a \emph{parallel batched} data structure instead of a concurrent one.
%
They provide provable bounds on the running time
of a dynamic multithreaded parallel program using $P$ processes
and a specified \emph{scheduler}.
%
The proposed scheduler extends the work-stealing scheduler by maintaining
separate \emph{batch} work-stealing deques that are accessed whenever processes have
operations to be performed on the abstract data type.
A process with a task to be performed on the data structure stores it  
in a \emph{request array} and tries to acquire
a global lock.
If succeeded, the process puts the task to perform the
batch update in its batch deque.
Then all the processes with requests in the request array
run the work-stealing routine on the batch deques
until there are no tasks left.
The idea of~\cite{agrawal2014provably} is similar to ours.
However, their algorithm is designed for systems with the fixed set of processes,
whereas we allow the processes to join and leave the execution.
From a more formal perspective, our goals are different:
we aim at improving the performance of a \emph{concurrent data structure}
while their goal was to establish bounds on the running time of a \emph{
parallel program in dynamic multithreading}.
Furthemore, implementing a concurrent data structure from
its parallel batched counterpart for dynamic multithreading
is only one of the applications
of our technique, as sketched in Section~\ref{subsec:batch}.
%


\section{Concluding remarks}
\label{sec:discussion}
Besides performance gains, parallel combining can potentially 
bring other interesting benefits. 

First, a parallel batched implementation is typically provided
with bounds on the running time.
The use of parallel combining might allow us to derive bounds
on the operations of resulting \emph{concurrent} data structures.
Consider, for example, a binary search tree.
To balance the tree, state-of-the-art
concurrent algorithms use the relaxed
AVL-scheme~\cite{bouge1998height}.
This scheme guarantees that the height of the tree 
never exceeds the contention level (the number of concurrent
operations) plus the logarithm of the tree size.
%
%
Applying parallel combining to a parallel batched binary search tree
(e.g.,~\cite{blelloch2016just}), we get a concurrent tree
with a strict logarithmic bound on the height.
\pk{should we say that no
    implementation guarantees that now?}
  
Second, the technique might enable the first ever concurrent
  implementation of certain data types, for example, a \emph{dynamic
    tree}~\cite{acar2017brief}.
  \pk{Why is it interesting? Anything to
  add here?} 



As shown in Section~\ref{sec:exp}, our concurrent priority
queue performs well compared to state-of-the-art algorithms.
A possible explanation is that the underlying parallel batched
implementation is designed for \emph{static} multithreading and, thus,
it has little synchronization overhead.
This might not be the case for implementations
based on \emph{dynamic} multithreading, where
the overhead induced by the scheduler can be much higher.
We intend to explore this distinction in the forthcoming work.

\bibliographystyle{plainurl}
\bibliography{references}

\newpage
\appendix
\section{Priority Queue with Parallel Combining}
\label{sec:priority}

\emph{Priority queue} is an abstract data type that maintains
an ordered multiset and supports two operations:
\begin{itemize}
\item \texttt{Insert}($v$)~--- inserts value $v$ into the set;
\item $v \leftarrow$ \texttt{ExtractMin}()~--- extracts the smallest value from the set.
\end{itemize}

\subsection{Batched Priority Queues: Review}

%
Several batched priority queues were
proposed in the literature for different parallel machines~\cite{jaja1992introduction}.

Pinotti and Pucci \cite{pinotti1991parallel} proposed
a batched priority queue for a $p$-processor CREW PRAM
implemented as a heap
each node of which contains $p$ values in sorted order:
only batches of size $p$ are accepted.

Deo and Prasad \cite{deo1992parallel} proposed
a similar batched priority queue for a $p$-processor
EREW PRAM
which can accept batches of any size not excewding $p$.
But the batch processing time is proportional to $p$.
For example, even if the batch consists of only one operations and
only one process is involved in the computation,
the execution still takes $O(p \log m)$ time on a queue of size $m$.

Brodal et al. \cite{brodal1998parallel} proposed
a batched priority queue that accepts
batches of \texttt{Insert} and \texttt{DecreaseKey} operations,
but not the batches of \texttt{ExtractMin} operations.
The priority queue maintains a set of
pairs $(\text{key}, \text{element})$ ordered by keys.
\texttt{Insert} operation takes a pair $(k, e)$ and inserts it.
\texttt{DecreaseKey} operations takes a pair $(d, e)$, searches in the
queue for a pair $(d', e)$ such that $d < d'$ and then replaces
$(d', e)$ with $(d, e)$.

Sanders \cite{sanders1998randomized} developed
a randomized distributed priority queue for MIMD.
MIMD computer has $p$ processing elements that 
communicate via asynchronous message-passing.
%
%
Again, the batch execution time is proportional to $p$,
regardless of the batch size.
%
%

Bingmann et al. \cite{bingmann2015bulk}
described a variation of the priority queue by Sanders~\cite{sanders1998randomized}
for external memory and, thus, it has the same issue.

To summarize, all earlier implementations we are aware of 
are tailored for a fixed number of processes $p$.
As a result, (1)~the running time of
the algorithms always depend on $p$, regardless of the
batch size and the number of involved processes;
(2)~once a data structure is constructed, we are unable to introduce more processes
into system and use them efficiently.
To respond to this issues, we propose a new parallel batched algorithm that applies a batch
of size $c$ to a queue of size $m$
in $O(c \cdot (\log c + \log m))$ RMRs
for CC or DSM models in total and in $O(c + \log m)$ RMRs per process.
By that, our algorithm can use up to $c \approx \log m$ processes efficiently.

\subsection{Sequential Binary Heap}

Our batched priority queue
is based on the \emph{sequential binary heap}
by Gonnet and Munro \cite{gonnet1986heaps},
one of the simplest and fastest sequential priority queues.
We briefly describe this algorithm below.

A binary heap of size $m$ is represented as a complete binary tree
with nodes indexed by $1, \ldots, m$.
Each node $v$ has at most two children: $2v$ and $2v + 1$ (to exist, $2v$ and $2v + 1$
should be less than or equal to $m$).
For each node, the \emph{heap property} should be satified:
the value stored at the node is less than the values stored at its children.

The heap is represented with \emph{size} $m$ and an array $a$
where $a[v]$ is the value at node $v$.
Operations \texttt{ExtractMin} and \texttt{Insert} are performed as follows:
\begin{itemize}
\item \texttt{ExtractMin} records the value $a[1]$ as a response,
  copies $a[m]$ to $a[1]$,
  decrements $m$ and performs the \textit{sift down} procedure to restore the heap property.
  Starting from the root, for each node $v$ on the path,
  we check whether value $a[v]$ is less than values
  $a[2v]$ and $a[2v + 1]$. If so, then the heap property is satisfied and
  we stop the operation.
  Otherwise, we choose the child $c$, either $2v$ or $2v + 1$, with the smallest value, swap values $a[v]$
  and $a[c]$, and continue with $c$.
\item \texttt{Insert($x$)} initializes a variable $val$ to $x$, increments $m$
  and traverses the path from the root to a new node $m$.
  For each node $v$ on the path, if $val<a[v]$, then the two values
  are swapped.
  Then the operation continues with the child of $v$ that lies on the path
  from $v$ to node $m$.
  Reaching node $m$ the operation sets its value to $val$.
\end{itemize}
The complexity is $O(\log m)$ steps per operation.

\subsection{Combiner and Client. Classes}

Now, we describe our novel parallel batched priority queue
in the form of \texttt{COMBINER\_CODE} and \texttt{CLIENT\_CODE} that fits
the parallel combining framework described in Section~\ref{sec:pc}. 
It is based on the sequential binary heap by Gonnet and Munro~\cite{gonnet1986heaps}.
The code of necessary classes is presented in Figure~\ref{fig:pq:classes},
\texttt{COMBINER\_CODE} is presented in Figure~\ref{fig:pq:combiner},
and \texttt{CLIENT\_CODE} is presented in Figure~\ref{fig:pq:client}.

\begin{figure}
\input{code/priority-classes.tex}
\caption{Parallel Priority Queue. Classes}
\label{fig:pq:classes}
\end{figure}

\begin{figure}
\input{code/priority-combiner.tex}
\caption{Parallel Priority Queue. \texttt{COMBINER\_CODE}}
\label{fig:pq:combiner}
\end{figure}

\begin{figure}
\input{code/priority-client.tex}
\caption{Parallel Priority Queue. \texttt{CLIENT\_CODE}}
\label{fig:pq:client}
\end{figure}

We introduce a sequential object \texttt{InsertSet}
(Figure~\ref{fig:pq:classes} Lines~\ref{line:classes:is:begin}-\ref{line:classes:is:end})
that consists
of two sorted linked lists $A$ and $B$
supporting \emph{size} operation $|\cdot|$.
The size of InsertSet $S$ is $|S| = |S.A| + |S.B|$.
InsertSet supports operation \texttt{split}:
\texttt{(X, Y) $\leftarrow$ S.split($\ell$)},
which splits InsertSet $S$ into two InsertSet objects $X$ and $Y$,
where $|X| = \ell$ and $|Y| = |S| - \ell$.
This operation is executed sequentially
in $O(L = \min(\ell, |S| - \ell))$ steps.
The split operation works as follows:
1)~new InsertSet $T$ is created (Line~\ref{line:classes:is:creation});
2)~if $|S.A| \geq L$ then the first $L$ values
  of $S.A$ are moved to $T.A$ (Lines~\ref{line:classes:is:amove:1}-\ref{line:classes:is:amove:2});
  otherwise, the first $L$ values from $S.B$ are moved to $T.B$
  (Lines~\ref{line:classes:is:bmove:1}-\ref{line:classes:is:bmove:2});
  note that either $|S.A|$ or $|S.B|$ should be at least $L$;
3)~if $L = \ell$ then $(T, S)$ is returned; otherwise, $(S, T)$ is returned
(Lines~\ref{line:classes:is:return}-\ref{line:classes:is:end}).

The heap is defined by its size $m$ and an array $a$ of Node objects.
Node object (Figure~\ref{fig:pq:classes} Lines~\ref{line:classes:node:1}-\ref{line:classes:node:2})
has three fields: value $val$, boolean $locked$ and InsertSet $split$.

\texttt{STATUS\_SET} consists of three items: \texttt{INITIAL}, \texttt{SIFT} and \texttt{FINISHED}.

A Request object (Figure~\ref{fig:pq:classes}
Lines~\ref{line:classes:request:1}-\ref{line:classes:request:2}) consists of:
a method \textit{method} to be called
and its input argument $v$;
a result \textit{res} field;
a \textit{status} field;
a node identifier $start$.

%

\subsection{ExtractMin Phase}

\subparagraph*{Combiner: ExtractMin preparation}
(Figure~\ref{fig:pq:combiner} Lines~\ref{line:pq:combiner:load}-\ref{line:pq:combiner:em:end}).
First, the combiner withdraws requests $A$
from the combining data structure $\mathbb{C}$ (Line~\ref{line:pq:combiner:load}).
If the size of $A$ is larger than $m$, 
the combiner serves the requests sequentially
(Lines~\ref{line:pq:combiner:seq:1}-\ref{line:pq:combiner:seq:2}).
Intuitively, in this case, there is no way to parallelize the execution.
For example, if $A$ consists of only Insert requests and if there are more Insert requests than the number of
nodes in the corresponding binary tree,
we cannot insert them in parallel. 

In the following, we assume that the size of the queue is at least
the size of $A$.
The combiner splits $A$ into sets $E$ and $I$
(Lines~\ref{line:pq:combiner:split:1}-\ref{line:pq:combiner:split:2}): the set of ExtractMin requests
and the set of Insert requests. 
Then it
finds $|E|$ nodes $v_1, \ldots, v_{|E|}$ of heap with the smallest
values (Lines~\ref{line:pq:combiner:best:1}-\ref{line:pq:combiner:best:2}), e.g.,
using the Dijkstra-like algorithm in $O(|E| \cdot \log |E|)$ primitive steps or $O(|E|)$ RMRs.
For each request $E[i]$, the combiner sets $E[i].res$ to $a[v_i].val$,
$a[v_i].locked$ to \texttt{true}, and $E[i].start$ to $v_i$
(Lines~\ref{line:pq:combiner:em:setup:1}-\ref{line:pq:combiner:em:setup:2}).

The combiner proceeds by pairing Insert requests in $I$
with ExtractMin requests in $E$ using the following procedure.
(Lines~\ref{line:pq:combiner:em:pairing:1}-\ref{line:pq:combiner:em:pairing:2}).
Suppose that $\ell = \min(|E|, |I|)$. For each $i \in [1, \ell]$,
the combiner sets $a[v_i].val$ to $I[i].v$
and $I[i].status$ to \texttt{FINISHED},
i.e., this Insert request becomes completed.
Then, for each $i \in [\ell + 1, |E|]$, the combiner sets $a[v_i].val$ to the value of the last node $a[m]$ and
decreases $m$, as in the sequential algorithm
(Lines~\ref{line:pq:combiner:em:setup2:1}-\ref{line:pq:combiner:em:setup2:2}).
Finally, the combiner sets the status of all requests in $E$ to \texttt{SIFT}
(Lines~\ref{line:pq:combiner:em:sift:1}-\ref{line:pq:combiner:em:sift:2}).

\subparagraph*{Clients: ExtractMin phase}
(Figure~\ref{fig:pq:client} Lines~\ref{line:pq:client:em:start}-\ref{line:pq:client:em:end}).
Briefly, the clients sift down the values in nodes $v_1, \ldots, v_{|E|}$
in parallel using hand-over-hand locking:
the $locked$ field of a node is set whenever there is a \emph{sift down} operation
working on that node.

A client $c$ waits until the status of its request becomes \texttt{SIFT}.
$c$ starts sifting down from $req.start$.
Suppose that $c$ is currently at node $v$.
$c$ waits until the $locked$ fields of the children become \texttt{false}
(Lines~\ref{line:pq:client:em:wait:1}-\ref{line:pq:client:em:wait:2}).
If $a[v].val$, the value of $v$, is less than the values
in its children, then \emph{sift down} is finished
(Lines~\ref{line:pq:client:em:stop:1}-\ref{line:pq:client:em:stop:2}):
$c$ unsets $a[v].locked$ and sets the status of its request to \texttt{FINISHED}.
Otherwise, let $w$ be the child with the smallest value.
Then $c$ swaps $a[v].val$ and $a[w].val$, sets $a[w].locked$,
unsets $a[v].locked$ and continues with node $w$
(Lines~\ref{line:pq:client:em:swap}-\ref{line:pq:client:em:end}).
%

If the request of the combiner is ExtractMin, it also runs the code above as a client
(Figure~\ref{fig:pq:combiner} Lines~\ref{line:pq:combiner:em:own:1}-\ref{line:pq:combiner:em:own:2}).
The combiner considers the ExtractMin phase completed when all requests in $E$ have status \texttt{FINISHED}
(Lines~\ref{line:pq:combiner:em:wait}-\ref{line:pq:combiner:em:end}).

\subsection{Insert Phase}
\subparagraph*{Combiner: Insert preparation} 
(Figure~\ref{fig:pq:combiner} Lines~\ref{line:pq:combiner:i:start}-\ref{line:pq:combiner:i:end}).
For Insert requests, the combiner removes all completed requests
from $I$ (Line~\ref{line:pq:combiner:i:start}).
%
Nodes $m + 1, \ldots, m + |I|$ have to be leaves,
because we assume that the size of $I$ is at most
the size of the queue.
We call these leaves \emph{target nodes}.
The combiner then finds all \emph{split nodes}:
nodes for which the subtrees of both children
contain at least one target node.
(See Figure~\ref{fig:tree} for an example of how target and split
nodes can be defined.)

\begin{wrapfigure}{L}{0.5\textwidth}
\begin{center}
\includegraphics[width=6cm]{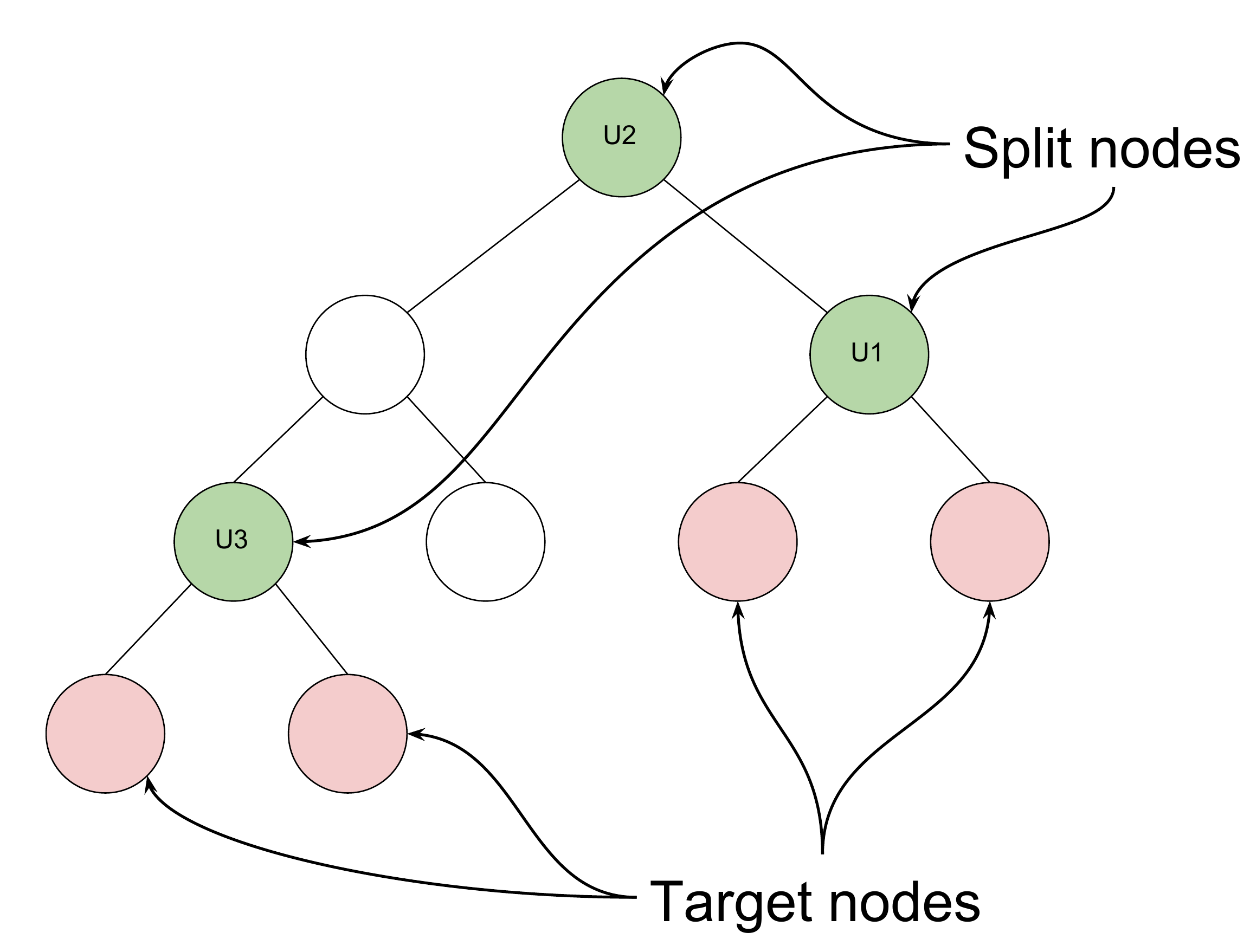}
\end{center}
\caption{Split and target nodes}
\label{fig:tree}
\end{wrapfigure}

Since we have $|I|$ target nodes,
there are exactly $|I| - 1$ split nodes $u_1, \ldots, u_{|I| - 1}$:
$u_i$ is the lowest common ancestor of nodes $m + i$ and
$m + i + 1$.
They can be found in $O(|I| + \log m)$ primitive steps
(Lines~\ref{line:pq:combiner:i:target:1}-\ref{line:pq:combiner:i:target:2}):
starting with node $m + i$ go up the heap
until a node becomes a left child of some node $pr$;
this $pr$ is $u_i$.
We omit the discussion about the fields $l$ and $r$ of $I[i]$:
they represent the smallest and the largest leaf identifiers
in the subtree of $u_i$, and they are used to calculate
the number of leaves that are newly inserted, i.e.,
$m + 1, \ldots, m + |I|$, in a subtree in constant time.
The combiner sets $I[1].start$ to the root (the node with identifier $1$),
(Line~\ref{line:pq:combiner:i:target:1}) and,
for each $i \in [2, |I|]$, it sets $I[i].start$ to
the right child of $u_{i - 1}$ (node $2 \cdot u_{i - 1} + 1$)
(Line~\ref{line:pq:combiner:i:setstart}).
Then the combiner creates an InsertSet object $X$, sorts the arguments
of the requests in $I$, puts them to $X.A$ and sets $a[1].split$ to $X$
(Lines~\ref{line:pq:combiner:i:is:1}-\ref{line:pq:combiner:i:is:2}).
Finally, it sets the status fields of all requests in $I$ to \texttt{SIFT}
(Lines~\ref{line:pq:combiner:i:sift:1}-\ref{line:pq:combiner:i:sift:2}).


\subparagraph*{Clients: Insert phase}
(Lines~\ref{line:pq:client:i:start}-\ref{line:pq:client:i:end}). 
Consider a client $c$ with an incompleted request $req$. 
It waits while $a[req.start].split$ is null
(Lines~\ref{line:pq:client:i:wait:1}-\ref{line:pq:client:i:wait:2}).
Now $c$ is going to insert values from InsertSet $a[req.start].split$ to the subtree of $req.start$.
Let $S$ be a local InsertSet variable initialized with $a[req.start].split$.
For each node $v$ on the path, 
$c$ inserts values from $S$ into the subtree of $v$.
%
%
$c$ calculates the minimum value $x$ in $S$
(Lines~\ref{line:pq:client:i:min:1}-\ref{line:pq:client:i:min:2}):
the first element of $S.A$ or the first element of $S.B$.
If $a[v].val$ is bigger than $x$, then the client removes $x$ from $S$,
appends $a[v].val$ to the end of $S.B$ and sets $a[v].val$ to $x$
(Lines~\ref{line:pq:client:i:swap:1}-\ref{line:pq:client:i:swap:2}).
Note that by the algorithm $S.B$ contains only values that were stored
in the nodes above node $v$, thus, any value in $S.B$ cannot be bigger than $a[v].val$
and after appending $a[v].val$ $S.B$ remains sorted.
%
Then the client calculates the number $inL$ of target nodes in the subtree of the left child of $v$ and
the number $inR$ of target nodes in the subtree of the right child of $v$
(Lines~\ref{line:pq:client:i:subtree:1}-\ref{line:pq:client:i:subtree:2}, to calculate
these numbers in constant time we use fields $l$ and $r$ of the request).
If $inL = 0$, then all the values in $S$ should be inserted into the
subtree of the right child of $v$,
and $c$ proceeds with the right child $2v + 1$
(Lines~\ref{line:pq:client:i:l0:1}-\ref{line:pq:client:i:l0:2}).
If $inR = 0$, then, symmetrically, $c$ proceeds with the left child $2v$
(Lines~\ref{line:pq:client:i:r0:1}-\ref{line:pq:client:i:r0:2}).
Otherwise, if $inL \ne 0$ and $inR \ne 0$, $v$ is a split node and, thus, there is a client that
waits at the right child $2v + 1$.
Hence, $c$ splits $S$ to \texttt{(X, Y) $\leftarrow$ S.split(inL)}
(Line~\ref{line:pq:client:i:split}): the
values in $X$ should be inserted into the
subtree of node $2v$ and the values in $Y$ should be inserted into the
subtree of node $2v+1$.
Then $c$ sets $a[2v + 1].split$ to $Y$, sets $S$ to $X$ and proceeds
to node $2v$ (Lines~\ref{line:pq:client:i:lr:1}-\ref{line:pq:client:i:lr:2}).
When $c$ reaches a leaf $v$ it sets the value $a[v].val$ to the only value in $S$
(Lines~\ref{line:pq:client:i:last}-\ref{line:pq:client:i:end}) and
sets the status of the request $req$ to \texttt{FINISHED} (Line~\ref{line:pq:client:finished}).

If the request of the combiner is an incompleted Insert, it runs the code
above as a client
(Figure~\ref{fig:pq:combiner} Lines~\ref{line:pq:combiner:i:own:1}-\ref{line:pq:combiner:i:own:2}).
The combiner considers the Insert phase completed
when all requests in $I$ have status \texttt{FINISHED}
(Lines~\ref{line:pq:combiner:i:wait}-\ref{line:pq:combiner:i:end}).

\subsection{Analysis}
Now we provide theorems on correctness and time bounds.
\begin{theorem}
Our concurrent priority queue implementation is linearizable.
\end{theorem}
\begin{proof}
The execution can be split into combining phases
(Figure~\ref{fig:pq:combiner} Lines~\ref{line:pq:combiner:load}-\ref{line:pq:combiner:i:end})
which do not intersect.
We group the operations into batches corresponding
to the combining phase in which they are applied.

Consider the $i$-th combining phase.
%
%
We linearize all the operations from the $i$-th phase
right after the end of the corresponding
\texttt{getRequests()} (Line~\ref{line:pq:combiner:load}) in the following
order: at first, we linearize ExtractMin operations
in the increasing order of their responses, then, we linearize
Insert operations in any order.

To see that this linearization is correct it is enough to prove that 
the combiner and the clients apply the batch correctly.
\begin{lemma}
Suppose that the batch of the $i$-th combining phase
contains $a$ ExtractMin operations and $b$ Insert operations
with arguments $x_1, \ldots, x_b$.
Let $V$ be the set of values stored in the priority queue
before the $i$-th phase.
The combiner and the clients apply this batch correctly:
\begin{itemize}
\item The minimal $a$ values $y_1, \ldots, y_a$ in $V$ are returned
  to ExtractMin operations.
\item After an execution the set of values stored in the queue
  is equal to $V \cup \{x_1, \ldots, x_b\} \setminus \{y_1, \ldots, y_a\}$.
  and the values are stored in nodes with identifiers $1, \ldots, |V| - b + a$.
\item After an execution the heap property is satisfied for each node.
\end{itemize}
\end{lemma}
\begin{proof}
The first statement is correct, because the combiner chooses the smallest
$a$ elements from the priority queue and sets them as the results of ExtractMin requests
(Lines~\ref{line:pq:combiner:best:1}-\ref{line:pq:combiner:em:setup:2}).

The second statement about the set of values straightforwardly follows
from the algorithm.
During ExtractMin phase the combiner finds $a$ smallest elements,
replaces them with $x_1, \ldots, x_{\min(a, b)}$
and with values from the last $a - \min(a, b)$ nodes of the heap:
the set of values in the priority queue becomes
$V \cup \{x_1, \ldots, x_{\min(a, b)}\} \setminus \{y_1, \ldots, y_b\}$ and
the values are stored in nodes $1, \ldots, |V| - a + \min(a, b)$.
Then, the \emph{sift down} is initiated, but it does not change the set of values and
it does not touch nodes other than $1, \ldots, |V| - a + \min(a, b)$.
During Insert phase the values $x_{\min(a, b) + 1}, \ldots, x_{b}$ are inserted and
new nodes which are used in Insert phase are $|V| - a + \min(a, b) + 1, \ldots, |V| - a + b$.
Thus, the final set of values is $V \cup \{x_1, \ldots, x_b\} \setminus \{y_1, \ldots, y_a\}$
and the values are stored in nodes $1, \ldots, |V| - a + b$.

The third statement is slightly tricky.
At first, the combiner finds 
$a$ smallest elements that should be removed and replaces them with $x_1, \ldots, x_{\min(a, b)}$ and 
with values from the last $a - \min(a, b)$ nodes of the heap.
Suppose that these $a$ smallest elements were
at nodes $v_1, \ldots, v_a$, sorted
by their depth (the length of the shortest path from the root) in non-increasing order.
These nodes form a connected subtree where $v_a$ is the root of the heap.
Suppose that they do not form a connected tree or $v_a$ is not the root of the heap.
Then there exists a node $v_i$ which parent $p$ is not $v_j$ for any $j$.
This means that the values in nodes $v_1, \ldots, v_a$ are not the smallest $a$ values:
by the heap property a value in $p$ is smaller than the value in $v_i$.

Now $a$ processes perform \emph{sift down} from the nodes $v_1, \ldots, v_a$.
We show that when a node $v$ is unlocked, i.e., its \texttt{locked}
field is set to false, the value at $v$ is the smallest value in the subtree of $v$.
This statement is enough to show that the heap property holds for all
nodes after ExtractMin phase, because at that point
all nodes are unlocked.
%

Consider an execution of \emph{sift down}.
We prove the statement by induction on
the number of unlock operations.
Base. No unlock happened and the statement is satisfied for all
unlocked nodes, i.e., all the nodes except for $v_1, \ldots, v_a$.
Transition. Let us look right before the $k$-th unlock:
the unlock of a node $v$.
The left child $l$ of $v$ should be unlocked
and, thus, $l$ contains a value
that is the smallest in its subtree.
The same statement holds for the right child $r$ of $v$.
$v$ chooses the smallest value between the value at $v$ and the values at $l$ and $r$.
This value is the smallest in the subtree of $v$.
Thus, the statement holds for $v$ when unlocked.

After that, the algorithm applies the incompleted $b - \min(a, b)$ Insert operations.
We name the nodes with at least one target node in the subtree
as \emph{modified}.
Modified nodes are the only nodes whose value can be changed
and, also, each modified node is visited by exactly one client.
To prove that after the execution the heap property
for each modified node holds: we show by induction
on the depth of a modified node that if a node $v$
is visited by a client with InsertSet $S$ then:
(1)~$S.A$ is sorted; (2)~$S.B$ is sorted and contains only values
that were stored in ancestors of $v$ after ExtractMin phase;
and (3)~$v$ contains the smallest value in its subtree when the client finishes with it.
Base. In the root $S.A$ is sorted, $S.B$ is empty and the
new value in the root is either the first value in $S.A$ or the current value in
the root, thus, it is the smallest value in the heap.
Transition from depth $k$ to depth $k + 1$.
Consider a modified node $v$ at depth $k + 1$ and
its parent $p$.
Suppose that $p$ was visited by a client with InsertSet $S_p$.
By induction, $S_p.A$ is sorted and $S_p.B$ is sorted and contains
only the values that were in ancestors.
Then the client chooses the smallest value in $p$: either $a[p]$,
the first value of $S_p.A$ or the first value of $S_p.B$.
Note that after any of these three cases $S_p.A$ and $S_p.B$ are sorted and
$S_p.B$ contains only values from ancestors and node $p$:
\begin{itemize}
\item $a[p]$ is the smallest, then $S_p.A$ and $S_p.B$ are not modified;
\item we poll the first element of $S_p.A$ or $S_p.B$; $S_p.A$ and $S_p.B$ are still sorted;
  then we append $a[p]$ to $S_p.B$, and $a[p]$ has to be the biggest element in $S_p.B$, since
  $S_p.B$ contains only the values from ancestors.
\end{itemize}
Then the client splits $S_p$ and some client, possibly, another one, works on $v$ with IntegerSet $S$.
Since, $S$ is a subset of $S_p$ then $S.A$ is sorted and $S.B$ is sorted and contains only the values
from ancestors (ancestors of $p$ and, possibly, $p$).
Finally, the client chooses the smallest value to appear in the subtree: the first value of $S.A$,
the first value of $S.B$ and $a[v]$.
\end{proof}

\end{proof}

\begin{theorem}
Suppose that the combiner collects $c$ requests using \texttt{getRequests()}.
Then the combiner and the clients
apply these requests to a priority queue of size $m$
using $O(c + \log m)$ RMRs in CC model each
and $O(c \cdot (\log c + \log m))$ RMRs in CC model in total.
\end{theorem}
\begin{proof}
Suppose that the batch consists of $a$ \texttt{ExtractMin} operations
and $b$ \texttt{Insert} operations.

The combiner splits requests into two sets $E$ and $I$ ($O(c)$ RMRs,
Lines~\ref{line:pq:combiner:split:1}-\ref{line:pq:combiner:split:2}).
Then it finds $a$ nodes with the smallest values ($O(a \log a)$ primitive steps, but
$O(a)$ RMRs, Lines~\ref{line:pq:combiner:best:1}-\ref{line:pq:combiner:best:2})
using Dijkstra-like algorithm.
After that, the combiner sets up ExtractMin requests, sets their status to \texttt{SIFT}
and pairs some Insert requests with ExtractMin requests
($O(a)$ RMRs, Lines~\ref{line:pq:combiner:em:setup:1}-\ref{line:pq:combiner:em:sift:2}).

The clients participate in ExtractMin phase.
At first, each client waits for its status to change ($1$ RMR).
Then the client performs at most $\log m$ iterations
of the loop (Line~\ref{line:pq:client:em:while}):
waits on the \textit{locked} fields of the children ($O(1)$ RMRs,
Lines~\ref{line:pq:client:em:wait:1}-\ref{line:pq:client:em:wait:2});
reads the values in the children ($O(1)$ RMRs,
Line~\ref{line:pq:client:em:load});
compares these values with the value at the node,
possibly, swap the values, lock the proper child and unlock the node ($O(1)$ RMRs,
Lines~\ref{line:pq:client:em:compare}-\ref{line:pq:client:em:end}).
When the client stops it changes the status ($1$ RMR, Line~\ref{line:pq:client:finished}).

The combiner waits for the change of the status of the clients ($O(a)$ RMRs,
Lines~\ref{line:pq:combiner:em:wait}-\ref{line:pq:combiner:em:end}).
Summing up, in ExtractMin phase each client performs $O(\log m)$ RMRs and
the combiner performs $O(a + \log m)$ RMRs, giving $O(c + c \cdot \log m)$ RMRs in total.

The combiner throws away completed Insert requests ($O(b)$ primitive steps and $0$ RMRs,
Line~\ref{line:pq:combiner:i:start}).
Then it finds the split nodes ($O(\log m + b)$ primitive steps, but $0$ RMRs,
Lines~\ref{line:pq:combiner:i:target:1}-\ref{line:pq:combiner:i:target:2}).
After that the combiner sorts arguments of remaining Insert requests,
sets their status to \texttt{SIFT} and
sets up the initial InsertSet
($O(b \cdot \log b)$ primitive steps, but $O(b)$ RMRs,
Line~\ref{line:pq:combiner:i:is:1} and Line~\ref{line:pq:combiner:i:sift:2}).

The clients participate in Insert phase.
At first, a client $t$ waits while the corresponding InsertSet is null ($1$ RMR,
Lines~\ref{line:pq:client:i:start}-\ref{line:pq:client:i:wait:2}).
Suppose that it reads the InsertSet $S$ and starts the traversal down.
The client performs at most $\log m$ iterations
of the loop (Line~\ref{line:pq:client:i:while}):
choose the smallest value ($O(1)$ RMRs,
Lines~\ref{line:pq:client:i:min:1}-\ref{line:pq:client:i:swap:2}),
find whether to split InsertSet ($O(1)$ RMRs,
Lines~\ref{line:pq:client:i:split:1}-\ref{line:pq:client:i:r0:2}),
split InsertSet (calculated below, Line~\ref{line:pq:client:i:split}) and
pass one InsertSet to another client ($O(1)$ RMRs,
Lines~\ref{line:pq:client:i:lr:1}-\ref{line:pq:client:i:lr:2}).
Now let us calculate the number of RMRs spent in Line~\ref{line:pq:client:i:split}.
Suppose that there are $k$ iterations of the loop and the size of $S$ at iteration $i$
is $s_i$.
At the $i$-th iteration split works in $O(\min(s_{i + 1}, s_i - s_{i + 1})) = O(s_i - s_{i + 1})$
primitive steps and RMRs.
Summing up through all iterations we get $O(s_1) = O(b)$ RMRs spent by $t$ in Line~\ref{line:pq:client:i:split}.
Finally, $t$ sets the value in the leaf ($O(1)$ RMRs, Lines~\ref{line:pq:client:i:last}-\ref{line:pq:client:i:end})
and changes the status ($O(1)$ RMRs, Line~\ref{line:pq:client:finished}).

The combiner waits for the change of the status of the clients ($O(b)$ RMRs,
Lines~\ref{line:pq:combiner:i:wait}-\ref{line:pq:combiner:i:end}).
Summing up, in Insert phase the clients and the combiner perform $O(b + \log m)$ RMRs each.
Consequently, the straightforward bound on the total number of RMRs is 
$O(c^2 + c \cdot \log m)$ RMRs.

To get the improved bound we carefully calculate
the total number of RMRs spent on the splits
of InsertSets in Line~\ref{line:pq:client:i:split}.
This number equals to the number of values that are moved
to newly created sets during the splits.
For simplicity, we assume that inserted values
are bigger than all the values in the priority queue and, thus,
each InsertSet contains only the newly inserted values.
This assumption does not affect the bound.
Consider now the inserted value $v$.
Suppose that $v$ was moved $k$ times and at the $i$-th time
it was moved during the split of InsertSet with size $s_i$.
Because $v$ is moved during split only to the set with the smaller size: $s_1 \geq 2 \cdot s_2 \geq
\ldots \geq 2^{k - 1} \cdot s_k$.
$k$ is less than $\log c$, because $s_1 \leq c$, and, thus, $v$ was moved no more than $\log c$ times.
This means, that in total during the splits of InsertSets
no more than $c \cdot \log c$ values are moved to new sets,
giving $O(c \cdot \log c)$ RMRs during the splits. 
This gives us a total bound of $O(c \cdot (\log c + \log m))$ RMRs during Insert phase.

To summarize, the combiner and the clients perform $O(c + \log m)$ RMRs each
and $O(c \cdot (\log c + \log m))$ RMRs in total.
\end{proof}

\begin{remark}
The above bounds also hold in DSM model for the version of the described algorithm.
For that we have to simply make spin-loops to loop on the local variable of processes.
In our algorithm the purpose of each spin-loops is to wake up some process.
At most places in our algorithm when we set the variable on which we spin we know
(or can deduce by a simple modification of the algorithm) which process is going to wake up.
For each spin-loop it is enough to create a separate variable in the memory of the target process.

The only two non-trivial spin-loops are in \texttt{CLIENT\_CODE} (Lines~\ref{line:pq:client:em:wait:1}-\ref{line:pq:client:em:wait:2})
where we do not know a process that is going to wake up.
To obviate this issue we expand each Node object with the pointer to process $proc$.
When the thread wants to \emph{sift-down}, first, it registers itself in $a[v].proc$ and, then,
checks $a[2v].locked$ and $a[2v+1].locked$.
If some of them are \texttt{true} then it spins on specifically created local variables:
on $notify_{2v}$ if $a[2v].locked$ is \texttt{true}, and on $notify_{2v + 1}$ if $a[2v+1].locked$ is \texttt{true}.
Then, the algorithm standardly performs swapping routine.
At the end, it unlocks the node, i.e., sets $a[v].locked$ to \texttt{false}, then, reads a process $a[v/2].proc$
and notifies it by setting its corresponding variable $notify_{v}$.
Note that the total number of \emph{notify} local variables
that is needed by each process is logarithmic from the size of the queue.

The described transformation (in reality, it is slightly more technical than described above) of our algorithm
provides an algorithm with the same bounds on RMRs but in DSM model.
\end{remark}

\end{document}